\documentclass[%
 reprint,
superscriptaddress, 
 amsmath,amssymb,
 aps,
prl
]{revtex4-2}

\usepackage{graphicx}
\usepackage{dcolumn}
\usepackage{bm}
\usepackage{braket}
\usepackage{mathrsfs}  
\DeclareMathOperator{\Tr}{Tr}
\usepackage{multirow}
\usepackage{relsize}
\usepackage[colorlinks=true, linkcolor=blue, citecolor=blue, urlcolor=black
]{hyperref}%
\usepackage{amsthm}
\usepackage{enumitem} 
\usepackage{xcolor}
\theoremstyle{definition}
\newtheorem{definition}{Definition}

\theoremstyle{plain}
\newtheorem{theorem}{Theorem}
\newtheorem{lemma}{Lemma}

\begin{document}

\preprint{APS/123-QED}



\title{How contextuality and antidistinguishability are related}

\author{Maiyuren Srikumar}
\email{msri0669@uni.sydney.edu.au}
\affiliation{Centre for Engineered Quantum Systems, School of Physics, University of Sydney, Sydney, NSW 2006, Australia.}
\affiliation{Sydney Quantum Academy, Sydney, NSW, Australia.}

\author{Stephen D. Bartlett}
\affiliation{Centre for Engineered Quantum Systems, School of Physics, University of Sydney, Sydney, NSW 2006, Australia.}

\author{Angela Karanjai}
\affiliation{Centre for Engineered Quantum Systems, School of Physics, University of Sydney, Sydney, NSW 2006, Australia.}

\date{\today}

\begin{abstract}

Contextuality is a key characteristic that separates quantum from classical phenomena and an important tool in understanding the potential advantage of quantum computation. 
However, when assessing the quantum resources available for quantum information processing, there is no formalism to determine whether a set of states can exhibit contextuality and whether such proofs of contextuality indicate anything about the resourcefulness of that set.  
Introducing a well-motivated notion of what it means for a set of states to be contextual, we establish a relationship between contextuality and antidistinguishability of sets of states. We go beyond the traditional notions of contextuality and antidistinguishability and treat both properties as resources, demonstrating that the degree of contextuality within a set of states has a direct connection to its level of antidistinguishability.
If a set of states is contextual, then it must be weakly antidistinguishable and vice-versa. However, critical contextuality emerges as a stronger property than traditional antidistinguishability.

\footnotetext{
$^\dag$Author@Author.unimelb.edu.au\\
$^\ddagger$Author@Author.edu.au\\
$^\S$Author@Author.edu.au}

\end{abstract}

\maketitle


\textit{Introduction.}--
Understanding what separates quantum from classical theory is central for the field of quantum computing. Although there have been many attempts to leverage insights from quantum foundations \cite{AnatoleKenfack_2004, PhysRevA.71.042302, Kleinmann_2011, 10.1088/1367-2630/16/1/013009, 10.1088/1367-2630/14/11/113011,10.1103/physrevlett.115.070501, 10.1103/physrevlett.119.050504, 10.1103/physrevx.8.011015, 10.48550/arxiv.1802.07744}, we are only beginning to understand how distinctly quantum features, such as contextuality, produce advantage.

Contextuality, which has its inception in the works of Bell \cite{RevModPhys.38.447}, Kochen and Specker \cite{10.1512/iumj.1968.17.17004}, captures the inability of explaining quantum mechanics through hidden variable theories that assign values to observables independently of other simultaneously measured properties. As it is an indicator of non-classical behaviour and subsumes well-known phenomena such as Bell non-locality and measurement incompatibility, contextuality has been identified as a resource in quantum information processing \cite{10.1038/nature13460, 10.1103/physrevlett.119.120505, 10.1103/physreva.95.052334,PhysRevA.101.012350, 10.1103/physreva.101.012120}.

Antidistinguishability was first introduced as \textit{Post-Peierls (PP) incompatibility} \cite{10.1103/physreva.66.062111} and plays a key role in the seminal Pusey-Barrett-Rudolph (PBR) Theorem \cite{10.1038/nphys2309}.
Antidistinguishability is an attribute given to a set of quantum states that is equivalent to the viability of \textit{unambiguous state exclusion} \cite{10.1103/physreva.89.022336},  where there exists a single measurement, for which each outcome conclusively rules out one state from the given set. 
This is in contrast to the well-known \textit{distinguishability} criterion that requires the existence of a measurement that determines the state with certainty. Though it is known that only orthogonal sets of states can be distinguishable \cite{Nielsen_Chuang_2010}, non-orthogonal states can be antidistinguishable~\cite{PhysRevResearch.5.023094}.

The concept of antidistinguishable states has been used in quantum communication tasks~\cite{10.1103/physrevlett.115.030504, 10.48550/arxiv.1903.04899, 10.1103/physrevresearch.2.013326, PhysRevA.109.032627} with demonstrations of exponential separation between classical and quantum communication~\cite{10.1103/physrevresearch.2.013326, 10.1103/physrevresearch.2.013326}, in quantum cryptography with the realisation of quantum digital signatures~\cite{10.1103/physrevlett.113.040502}, and more recently in demonstrating polynomial separations in memory requirements in quantum machine learning algorithms \cite{10.48550/arxiv.2402.08606}. In quantum resource theories it has been demonstrated that certain measures of resources can be identified with advantages in state exclusion tasks \cite{10.1103/physrevlett.125.110401, 10.1103/physrevlett.125.110402}. From a foundational perspective, antidistinguishability has played a critical role in the debate between epistemic and ontic interpretations of the quantum state~\cite{10.1038/nphys2309, 10.12743/quanta.v3i1.22, 10.1103/physrevlett.112.250403, 10.48550/arxiv.2401.17980}.
Many of these works show that antidistinguishability has the ability to demonstrate differences between aspects of quantum and classical theory, leading to research into the conditions under which a set of states is antidistinguishable~\cite{10.1103/physreva.66.062111, 10.1103/physreva.89.022336, 10.1103/physrevresearch.2.013326, 10.1103/physreva.107.l030202, 10.48550/arxiv.1804.10457, 10.48550/arxiv.2309.03723, 10.48550/arxiv.2311.17047}.

Both contextuality and antidistinguishability reveal non-classical characteristics of quantum theory, however a direct link has remained elusive.
A possible connection was first suggested in Ref.~\cite{10.48550/arxiv.1802.07744} through the concept of a \textit{partitioning measurement}; such a measurement was guaranteed to exist when contextuality is present. 
More recently, aspects of the relationship have been further developed by establishing a method that employs antidistinguishability to construct noncontextual inequalities~\cite{10.1103/physreva.101.062113}, through finding pairwise \textit{antisets} that require triples of states to be antidistinguishable. 

In this Letter, we show an explicit equivalence between the contextuality of a set of states and weak antidistinguishability. This notion of weak antidistinguishability is related to \textit{weak conclusive 1-state exclusion} which has recently been used to characterise the Choi-rank of quantum channels \cite{stratton2024operationalinterpretationchoirank}. The connection between contextuality and weak antidistinguishability found in this letter therefore offers a possible link between contextuality and channel complexity.
Considering Kochen-Specker (KS) proofs of contextuality in the hypergraph framework of contextuality scenarios \cite{10.1103/physrevlett.112.040401, 10.1007/s00220-014-2260-1}, we show that a set of projectors that generate a KS proof of contextuality must necessarily be weakly antidistinguishable and a set of weakly antidistinguishable states must exhibit contextuality. 
We then examine the extent of contextuality of a set of states by looking at sets where removing just one state causes the remaining set to become non-contextual; we refer to such a set as being \emph{critically contextual}. Critical contextuality emerges as a much stronger criterion than antidistinguishability:  all critically contextual sets are antidistinguishable, however we find examples of antidistinguishable sets that are not critically contextual.  We refer to the stronger notion of antidistinguishability as strong antidistinguishability and conjecture that it implies critical contextuality. 
By establishing an equivalence between two well-studied concepts with a previously unknown relationship, we facilitate the exploration of one concept with the tools employed in the other.

\textit{Results.}--
We start by defining antidistinguishability with projective measurements. A projective measurement is represented by a projection-valued measure (PVM). Operationally, a PVM corresponds to a measurement with each projector corresponding to an outcome of the measurement. (We exclude PVMs that include the zero projector, as, despite being mathematically valid, the zero projector cannot physically be associated with a measurement outcome of any projective measurement.)

\begin{definition}\label{def:antidist}
     A finite set of states $\mathcal{S} = \{\rho_i\}_{i=1}^N$ is,
     \begin{enumerate}[label=(\roman*)]
        \item \textit{Weakly Antidistinguishable} (\textbf{WA}) if there exists a PVM $\{\pi_j\}_{j}$ such that, for each $\pi_j$,  $\exists \rho'\in \mathcal{S}$ such that $\Tr(\pi_j \rho')=0$,
         \item \textit{Antidistinguishable} (\textbf{A}) if there exists a PVM $\{\pi_i\}_{i=1}^N$ such that, $\Tr \left( \pi_i \rho_i \right) = 0$ for all $i=1,...,N$, 
         \item \textit{Strongly Antidistinguishable} (\textbf{SA}) if there exists a PVM $\{\pi_j\}_j$ that (1)~weakly antidistinguishes $\mathcal{S}$ and (2)~for each $\rho_i$ there exists a $\pi_i$ such that $\Tr \left( \pi_i \rho_i \right) = 0$ and $\Tr \left( \pi_i \rho_j \right) \neq 0$ for all $i\neq j$. 
     \end{enumerate}
\end{definition}
When a measurement is performed on an unknown state from a set resulting in a specific measurement outcome, any state(s) within that set that is orthogonal to the projector associated with the outcome can be excluded as the initial state. Thus we say that the particular outcome \emph{excludes} the state(s). 
A set of states is weakly antidistinguishable, if there exists a measurement -- a \emph{weakly antidistinguishing measurement} (\textbf{WA-PVM}) -- where every outcome excludes at least one state from the set.
Antidistinguishability of a set of states is a stronger criteria as it requires that there exist a measurement -- an \emph{antidistinguishing measurement} (\textbf{A-PVM}) -- where the number of outcomes is the same as the number of states in the set, and in addition, each state in the set has at least one outcome that excludes it.
In order for a set of states to be strongly antidistinguishable, there must exist a \emph{strongly antidistinguishing measurement} (\textbf{SA-PVM}) such that, for each state, there exists an outcome that exclusively excludes that state and no other state in the set.

From the definition of \textbf{SA}, one can see that if the PVM $\Pi=\{\pi_j\}_{j=1}^M$ is an \textbf{SA-PVM} for a set $\mathcal{S} = \{\rho_i\}_{i=1}^N$, $M\geq N$. It is always possible to combine the PVM elements in $\Pi$ to produce a new \textbf{A-PVM} that has the same cardinality as the set. This can be done as follows:\\
Let $\gamma_i=\sum_{\pi_j\in \Pi_i}\pi_j$, where $\Pi_{i}\subset \Pi$ such that $\Tr(\pi_j\rho_i)=0$ and $\Tr(\pi_j\rho_k)\neq 0$ for $k<i$ for all $\pi_j\in \Pi_i$. $\Gamma=\{\gamma_i\}_{i=1}^{N}$ is now an \textbf{A-PVM}, which shows that a strongly antidistinguishable set of states is necessarily antidistinguishable.
We therefore have the following chain of implications for a set of states:
\begin{equation*}
    \textbf{SA} \implies \textbf{A}  \implies \textbf{WA}
\end{equation*}
Antidistinguishability and strong antidistinguishability are, in general, sensitive to the addition of a state to the set, while the weaker attribute is unaffected. 

We now define the characterisation of contextuality to a set of states, similar to how antidistinguishability is a property of a set of states.  Recall that a set of observables is noncontextual if there exists a simultaneous assignment of outcomes to the observables such that all the functional relationships between commuting observables are satisfied. The set of observables is contextual if such a value assignment is impossible. We extend this notion to a set of states, as follows. Kochen and Specker's proof~\cite{10.1512/iumj.1968.17.17004} highlights contextuality through logical inconsistencies when assigning truth values to sets of measurement outcomes which are represented as sets of orthogonal projectors.  A pair of mutually orthogonal projectors represents a pair of mutually exclusive measurement outcomes and therefore cannot simultaneously be assigned a truth value of 1. We explore such constraints on truth value assignments of states within the hypergraph framework. Within this framework we start by describing contextuality formally using \textit{contextuality scenarios} and subsequently construct \textit{contextual instances} of scenarios,  that extend contextuality to sets of states.

In a $d$-dimensional Hilbert space $\mathcal{H}$, a set of rank-$1$ projectors $\{\sigma_1 , \sigma_2, ..., \sigma_{d}\}$, are considered a \textit{context} if it satisfies the following two criteria \cite{RevModPhys.94.045007}: (i) $\sigma_i \sigma_j = 0$ for all $i\neq j$ \textit{(orthogonality condition)} and (ii) $\sum_{i=1}^{d} \sigma_i = \mathbb{I}$ \textit{(completeness condition)}.
Physically, a context represents the (mutually exclusive) outcomes of a projective measurement. 
The set of all rank-$1$ projectors for a given Hilbert space is denoted as $\mathbb{X}(\mathcal{H})$ and we denote the set of all contexts formed from a given set of projectors $X\subseteq \mathbb{X}$ as $\mathbb{C}(X)$. Any set of projectors $X\subseteq \mathbb{X}(\mathcal{H})$ and their associated contexts $\mathcal{C} \subseteq \mathbb{C}(X)$ made entirely of projectors from $X$,  can be identified with a hypergraph $\mathcal{G}(X, \mathcal{C})$, with vertices being projectors and hyperedges denoting the contexts \cite{10.1103/physrevlett.112.040401, 10.1007/s00220-014-2260-1}.

\begin{definition}
(Ref.~\cite{ 10.1007/s00220-014-2260-1}.)  A contextuality scenario is a hypergraph $\mathcal{G}(X, \mathcal{C})$ with vertices $X$ corresponding to projectors and a set of hyperedges $\mathcal{C}\subseteq \mathbb{C}(X)$.
\end{definition}
Here we have defined contexts as rank-1 PVMs, which suffices for our results, though contexts can generally include higher-rank PVMs. For a more detailed discussion, refer to Appendix \ref{appen:rank_1}.

The construction of a contextuality scenario -- which we will often simple refer to as scenario -- does not necessarily indicate a proof of contextuality. Instead, a KS proof arises from not being able to assign truth values to the projectors (vertices) of a particular graph, independent of context, such that no mutually orthogonal pair of projectors are both assigned a truth value of 1. We can therefore define a contextuality scenario as being \textit{contextual} in the following way:
\begin{definition}\label{def:contextual}
A contextuality scenario $\mathcal{G}(X, \mathcal{C})$ is \textit{noncontextual} if there exists value assignment $\nu : X \rightarrow \{0,1\}$ such that, 
    \begin{align}\label{eq:val_assign_def}
    \sum_{\sigma_i\in C} \nu (\sigma_i) = 1  \end{align}
for all contexts $C\in \mathcal{C}$. Otherwise it is \textit{contextual}.
\end{definition}

Since each context corresponds to the set of distinct outcomes of a measurement, Eq.~\eqref{eq:val_assign_def} is the requirement that any measurement can only have one outcome.
Any value assignment that obeys \eqref{eq:val_assign_def} for all contexts in $\mathcal{C}$ is referred to as a \emph{noncontextual value assignment}. The existence of a noncontextual value assignment for $\mathcal{G}(X, \mathcal{C})$, corresponds to assigning two colours representing truth values to all vertices in a graph, with the constraint that vertices in a common hyperedge must have one, and only one, true assignment (i.e., only one vertex assigned the value of $1$). This satisfiability task is referred to as the \textit{KS colorability problem} \cite{alma991019892499703276}. Hence, in this picture, a contextuality scenario is contextual if it does not admit a KS-colouring of the graph $\mathcal{G}(X, \mathcal{C})$. Referring to a scenario as contextual in this way is a form of \textit{state-independent} contextuality. It is further possible to fix the value of a single projector (state) and ask whether it possible to assign values to the remainder of the graph, giving a \textit{state-dependent} form of contextuality. 

We extend the use of the orthogonality relationships between projectors to explore its implications for sets of states. The possibilitistic implications of the state of a system prior to measurement are captured by its orthogonality relationships with the outcomes of the measurement. For example, the outcome of a measurement on the system corresponding to a projector $\rho_{i}$ implies that the state of the system prior to the measurement could not have been orthogonal to $\rho_{i}$. Thus if one assigns a truth value of 1 to a state, then truth values of all projectors orthogonal to the state must be 0. We define a set of states as noncontextual if after assigning a truth value of 1 to all the states in the set, there exists a simultaneous truth value assignment to all its orthogonal projectors, which respects the orthogonality relationships. If such a value assignment is impossible, then the set of states is contextual.
Thus we are concerned with projectors that have implied truth value assignments as a result of assigning $1$ to a select set of projectors $\mathcal{S}$. 
More specifically, given a set of non-orthogonal projectors $\mathcal{S}=\{\rho_i\}$, we define its \textit{implied projectors} as, 
    \begin{equation}
        \mathcal{I}({\mathcal{S}}) = \{ \sigma\in \mathbb{X}(\mathcal{H}) \vert \exists \rho_i \in \mathcal{S} : \rho_i \sigma=0\} \,.
    \end{equation}
We can now define a contextuality scenario that is constructed through fixed and implied projectors. 
\begin{definition}
    A contextuality scenario $\mathcal{G}(X_\mathcal{S}, \mathcal{C}_\mathcal{S})$ is \textit{generated} (or \textit{implied}) by a set of projectors $\mathcal{S}$ if the following conditions are satisfied: (i) $\mathcal{S} \subset X_{\mathcal{S}}$, (ii) $X_{\mathcal{S}}\backslash \mathcal{S} \subset \mathcal{I}(\mathcal{S})$, and (iii) $\mathcal{C}_{\mathcal{S}}=\mathbb{C}(X_{\mathcal{S}})$. 
\end{definition}
\noindent In other words, $\mathcal{S}$ generates $\mathcal{G}(X_\mathcal{S}, \mathcal{C}_\mathcal{S})$ if all projectors in the hypergraph are either in $\mathcal{S}$ or orthogonal to at least one projector in $\mathcal{S}$ and all possible contexts between projectors in $\mathcal{S}$ is included in the scenario.
For example, the orange projectors in Figure \ref{fig:18_vec_proof} imply all other vertices in the graph and hence generate the 18-vector contextuality scenario presented. 
\begin{definition}\label{def:contextual_scenario}
     A non-orthogonal set of projectors $\mathcal{S}\subset X$ forms a \textit{contextual instance} of the implied scenario $\mathcal{G}(X_\mathcal{S}, \mathcal{C}_\mathcal{S})$ if there exist no noncontextual value assignment $\nu$ such that $\nu(\rho_i)=1$ for all $\rho_i \in \mathcal{S}$.
\end{definition}
\noindent As noted above, the contextual attribute of the scenario is stronger than that of the instance, meaning $\mathcal{G}(X_\mathcal{S}, \mathcal{C}_\mathcal{S})$ may admit a noncontextual value assignment while $\mathcal{S}$ forms a contextual instance. Nevertheless, a contextual instance allows us to formally characterise the contextuality of a given set of states.

\begin{definition}
    A set of pure states $\mathcal{S}=\{\rho_i\}$ is \textit{contextual} if there exists an implied contextuality scenario for which $\mathcal{S}$ is a contextual instance.
\end{definition} 
An example of a contextual set are the states associated with the orange triangles in Figure \ref{fig:18_vec_proof}, generating the 18-vector Kochen-Specker proof of contextuality \cite{Cabello_1996}.

Having defined contextuality and antidistinguishability for a set of states, we can now establish a simple equivalence between the two non-classical properties.

\begin{theorem}\label{theorem:main_proj}
    Let $\mathcal{S}=\{\rho_i\}$ be a set of non-orthogonal pure states. The set $\mathcal{S}$ is  weakly antidistinguishable if and only if $\mathcal{S}$ is contextual.
\end{theorem}
\begin{proof} 
See Appendix \ref{app:main_proj}.
\end{proof}

In order to establish the connection between contextuality and the standard notion of antidistinguishability, we need a stricter form of  contextuality for a set of states. 
\begin{definition}
    A set of pure states $\mathcal{S}=\{\rho_i\}$ is \textit{critically contextual} if there does not exist a contextual $\mathcal{S}'\subset \mathcal{S}$.
\end{definition}
We find that for any scenario $\mathcal{G}$ generated by a critically contextual set of states, every context in $\mathcal{G}$ that does not include states in $S$, must strongly antidistinguish $S$. This follows easily from the following Lemma.
\begin{lemma}\label{lemma:maxcontextual_pvm}
    Any PVM $\Pi=\{\pi_i\}$ which weakly antidistinguishes a critically contextual set $\mathcal {S}$ must also strongly antidistinguish it.
\end{lemma}

\begin{proof}
See Appendix \ref{app:lemma1}.
\end{proof}

We can now state the following important implication: 
\begin{theorem}\label{theorem:max_cont_antidist}
    A non-orthogonal set of pure states is critically contextual only if it is strongly antidistinguishable and therefore antidistingiushable.
\end{theorem}
\begin{proof}
    We know from Theorem \ref{theorem:main_proj} that a contextual (and thereby critically contextual) set of states must have a \textbf{WA-PVM} $\Pi$. By Lemma \ref{lemma:maxcontextual_pvm} we know that $\Pi$ must be strongly antidistinguishing. Since strongly antidistinguishable states are also antidistinguishable, we have proven the result. 
\end{proof}

\begin{figure}
    \centering    \includegraphics[width=220pt]{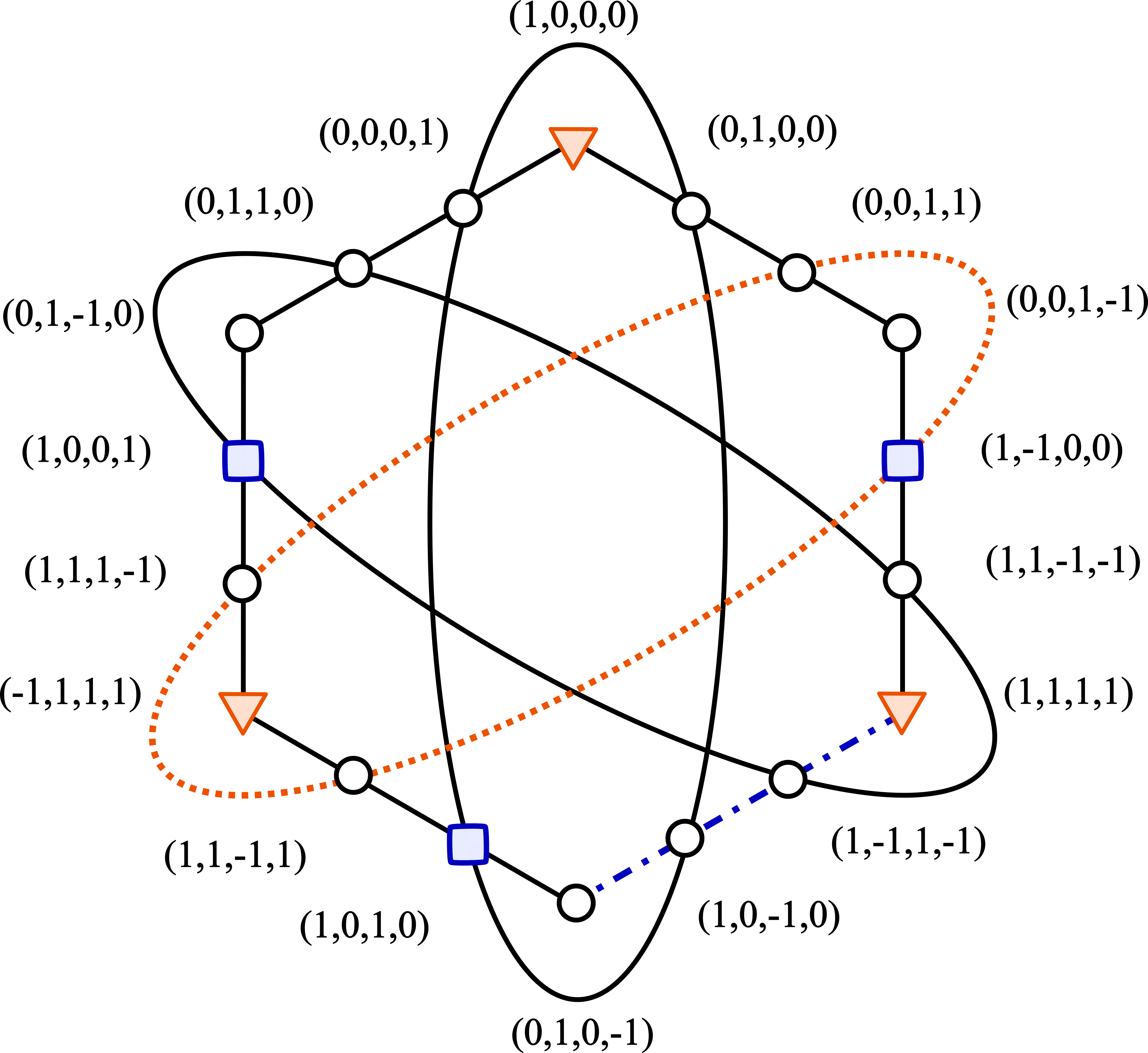}
    \caption{A contextual scenario with 18 vectors and 9 contexts in dimension $d=4$ \cite{Cabello_1996}. The vertices are projectors (pure states) where $(a,b,c,d)$ is identified with the projector $\ket{a,b,c,d}\bra{a,b,c,d}$ where $\ket{a,b,c,d}=a\ket{00}+b\ket{01}+c\ket{10}+d\ket{11}$. The hyperedges form contexts of mutually orthogonal sets, where for simplicity the hyperedges are illustrated as straight lines and ovals. The scenario is generated by projectors indicated by orange triangles, along with its \textbf{SA} and \textbf{A} context given by the orange, dotted line. The projectors indicated by blue squares do not generate the entire graph -- all but the $(1,0,0,0)$ state -- and therefore demonstrate only a contextual instance with the associated (strongly) antidistinguishing context given by the blue, dash-dotted line.}
    \label{fig:18_vec_proof}
\end{figure}

The relation between contextuality and antidistinguishability is demonstrated in Figure \ref{fig:18_vec_proof}, where we see that any set of projectors that generate a contextual scenario (for example the projectors indicated by orange triangles) must necessarily be at least \textbf{WA} as claimed in Theorem \ref{theorem:main_proj}. 
Hence, choosing a generating set of projectors from a known contextual proof guarantees the existence of a \textbf{WA} measurement. Furthermore, the removal of any one projector allows for a noncontextual scenario as two non-orthogonal states cannot be \textbf{WA} (otherwise they would be distinguishable) and thereby contextual from Theorem \ref{theorem:main_proj}. Hence we have a critically contextual set of states which, by Theorem \ref{theorem:max_cont_antidist}, are both \textbf{A} and \textbf{SA}.

In the example above, we have used a known contextual scenario to obtain an antidistinguishable set of states. We also provide an example using the \textit{PBR states} \cite{10.1038/nphys2309} where antidistinguishability can be used to establish the contextuality of a set of states, in Appendix \ref{app:pbr}.

The statement of Lemma \ref{lemma:maxcontextual_pvm} illustrates the strength of critical contextuality as a condition. For a set of states, if one finds a \textbf{WA-PVM} that is not also an \textbf{SA-PVM}, then one can conclude that the set of states is not critically contextual. This method of ruling out a set of states as being critically contextual is potentially much more efficient than considering all possible contextuality instances for a set of states and checking if they remain contextual after the removal of a state from the set. However, due to the one way implication of critical contextuality implying strong antidistinguishability, the existence of an \textbf{SA-PVM} cannot be used to confirm that a set of states is critically contextual. 
Critical contextuality emerges as a stronger property than antidistinguishability.
This can be seen with the example given by a set of states associated with the projectors labelled by blue squares in Figure \ref{fig:18_vec_proof} with the addition of the state associated with $(1,0,0,0)$. The blue dash-dotted hyperedge is therefore an antidistinguishing measurement for this larger set (note the orthogonality between $(1,0,0,0)$ and $(0,1,0,-1)$ projectors), while also being an antidistinguishing measurement for the subset of only those states labelled by blue squares. Theorem \ref{theorem:main_proj} states that the states labelled by blue squares will therefore be contextual, implying that the initial larger antidistinguishable set is not critically contextual. This example shows that antidistinguishability does not imply critical contextuality, but does not rule out strong antidistinguishability implying critical contextuality. We have not found a counterexample to strong antidistinguishability implying critical contextuality. For example, the PBR measurement is an \textbf{SA-PVM} for the PBR states, and, as we have shown, the PBR states also turn out to be critically contextual. Similarly, the Lison\v{e}k set \cite{PhysRevA.89.042101} has a \textbf{SA} set that is also critically contextual~\footnote{See Supplemental Material at [] for additional examples including the Lison\v{e}k set~\cite{PhysRevA.89.042101} and the Yu-Oh set~\cite{10.1103/physrevlett.108.030402}, and which includes Refs.~\cite{10.48550/arxiv.2408.16764}.}. Hence we conjecture, an equivalence between a set being \textbf{SA} and critically contextual.

\textit{Discussion.}--
We have defined a natural characterisation of contextuality for a set of pure states $\mathcal{S}$, as the inability to construct a noncontextual value assignment for a contextuality scenario generated by $\mathcal{S}$. It is worth noting that contextual sets of states need not be KS sets themselves; instead their significance lies in their capacity to generate KS sets.
The understanding of contextuality as a property of a set of states has implications for understanding contextuality beyond sequential measurements and as a resource that provides exclusionary measurements.
Defining contextuality as a resource for a set requires the natural supposition that contextuality should not disappear with the addition of another state. This property is satisfied by Def. \ref{def:contextual} and it is therefore consistent that Theorem \ref{theorem:main_proj} demonstrates a simple equivalence between the attributes of contextuality and weak antidistinguishability. Conversely, both \textbf{A} and \textbf{SA} properties of a set are sensitive to the addition of states. Hence we are able show that \textbf{A} and \textbf{SA} are related to the similarly sensitive critically contextual attribute in Theorem \ref{theorem:max_cont_antidist}. 

Here we have specifically focused on KS contextuality but in fact there exists more general notions of contextuality that do not present themselves as the violation of Eq.~\eqref{eq:val_assign_def}. For example, the Yu-Oh set \cite{10.1103/physrevlett.108.030402} is contextual with regards to a more general form of contextuality~\footnotemark[\value{footnote}]. 
Nevertheless, explicitly demonstrating KS sets has become increasingly important in recent literature. In Ref.~\cite{10.1103/physrevlett.134.010201}, it was shown that every bipartite perfect quantum strategy defines a KS set. KS sets also underpin many interesting bipartite quantum correlations, as demonstrated in Ref.~\cite{PhysRevResearch.6.L042035}. Furthermore, exploring contextuality as a resource for quantum computation \cite{PhysRevLett.102.050502, PhysRevA.88.022322, 10.1103/physrevlett.119.120505}, connecting KS sets to channel capacity \cite{PhysRevLett.104.230503}, all can be seen through KS contextuality. Our work supports these directions by showing that antidistinguishability can be used to identify generators of KS sets.

Although it is more natural to look at the way contextuality implies the existence of antidistinguishing measurements, the reverse is also useful. It has been shown that antidistinguishability characteristics can yield tighter noncontexuality inequalities~\cite{10.1103/physreva.101.062113}. We have shown that it is not possible to have a contextuality proof where the generating set of projectors are not at least \textbf{WA}.
Therefore the characteristics that determine a set of states to not be \textbf{A} can be used in understanding when contextual scenarios are possible. The simplest to consider is one of the first conditions for antidistinguishability for the specific case of three states \cite{10.1103/physreva.66.062111} (see Lemma \ref{lemma:cond_three_states_anti}). One consequence of such constraints is that a set of three projectors that have large overlaps (where the sum of pairwise overlaps is greater than one) cannot generate a scenario that demonstrates a proof of contextuality.  This result proves that the PBR set of states \cite{10.1038/nphys2309} is critically contextual.

Additional conditions have been established for states that are antidistinguishable, including more general bounds on the sum of overlaps \cite{10.1103/physreva.89.022336}, conditions on the spectral decomposition of the antidistinguishing measurement and states \cite{10.48550/arxiv.1804.10457}, and recently the incoherence of the Gram matrix associated with the set of states \cite{10.48550/arxiv.2311.17047}. More generally, a semi-definite program can yield the antidistinguishing measurement~\cite{10.1103/physreva.89.022336}, when possible. 
Therefore the inability to obtain an antidistinguishing solution can also be numerically determined, though not necessarily efficiently. These tools from antidistinguishability literature offer a clear method for identifying when projectors can generate contextuality proofs.

The concept of weak antidistinguishability is closely related to weak state exclusion, which has gained significance due to its recently established operational connection to the Choi-rank of a quantum channel \cite{stratton2024operationalinterpretationchoirank}. In particular, it has been shown that the Choi-rank bounds the success of an entanglement-assisted state exclusion task. This naturally prompts consideration of how our results might inform or extend this finding. Focusing on pure states, we identify a potential relationship between the Choi-rank and the contextuality exhibited by the set of states that have undergone evolution through the channel. A detailed exploration of this connection is left for future work.

\textit{Acknowledgements -- }This work is supported by the Australian Research Council via Discovery Project number DP220101771 and the Centre of Excellence in Engineered Quantum Systems (EQUS) project number CE170100009.

\appendix
\setcounter{secnumdepth}{5}

\section{Comment on Rank-1 Projectors}\label{appen:rank_1}

We have chosen to define contexts in the hypergraphs as rank-1 PVMs, however contexts can in general be composed of higher rank PVMs. For the results in this paper, we do not lose generality by making this choice.
This is because fine-graining of a PVM to a rank-1 PVM cannot remove contextuality, that is, if a higher rank PVM forms a contextual instance for a set of states, then any fine-graining of that PVM will also form a contextual instance for the same set. Thus when looking for the existence of a contextual instance for a set of states, it is sufficient to restrict to considering only rank-1 projectors.

It is, important to note that while fine-graining of a PVM  cannot remove contextuality, it may add contextuality. That is, given that a higher rank PVM does not form a contextual instance for a set of states, it is possible that, for some choice of fine-graining, the rank-1 PVM that it is decomposed into, may form a contextual instance for the same state of states. Thus in instances, where one wants to investigate whether a set of states is contextual or not with respect to a restricted set of measurements, consisting of higher rank PVMs, one can relax the condition of only rank-1 projectors forming a context in the hypergraph. One can then represent higher rank projectors as vertices of the hypergraph and define a context as a set of vertices that satisfies the orthogonality and completeness conditions.

\section{Proof of Theorem \ref{theorem:main_proj}}\label{app:main_proj}

\noindent
\textbf{Theorem 1.} \textit{Let $\mathcal{S}=\{\rho_i\}$ be a set of non-orthogonal pure states. The set $\mathcal{S}$ is  weakly antidistinguishable if and only if $\mathcal{S}$ is contextual.}

\begin{proof} 
We start with the \textit{if} direction by assuming we have a contextual set $\mathcal{S}$, i.e. we assume the existence of an implied scenario $\mathcal{G}(X_{\mathcal{S}}, \mathcal{C}_{\mathcal{S}})$ for which $\mathcal{S}$ is a contextual instance. Let $\nu:X_{\mathcal{S}} \rightarrow \{0,1\}$ be a value assignment such that, $\nu(\sigma)=1$ for all $\sigma \in \mathcal{S}$ and $\nu(\sigma_{i})\cdot\nu(\sigma_{j})=0$ if $\Tr[\sigma_{i}\sigma_{j}]=0$ for $\sigma_{i},\sigma_{j} \in \mathcal{S}$.
Since  $\mathcal{S}$ is a contextual instance for  $\mathcal{G}(X_{\mathcal{S}},\mathcal{C}_{\mathcal{S}})$, we know that $\nu$ cannot be noncontextual. 
Therefore, because each context can contain at most one projector in $\mathcal{S}$ -- otherwise states in $\mathcal{S}$ would be orthogonal -- there exists a context $\widetilde{C} \in \mathcal{C}_{\mathcal{S}}$ such that,
 \begin{equation}
 \sum_{\sigma\in \widetilde{C}} \nu(\sigma) = 0    
 \end{equation}
where we note that from the definition of $\nu$ we subsequently have that $\mathcal{S}\cap\widetilde{C}=\emptyset$.
Since $\mathcal{G}(X_{\mathcal{S}},\mathcal{C}_{\mathcal{S}}) $ is generated by $\mathcal{S}$, we have that every projector in $\widetilde{C}$ is orthogonal to at least one element in $\mathcal{S}$, that is, $\widetilde{C}\subset \mathcal{I}(\mathcal{S})$. We also have that because $\widetilde{C}$ is a context, it is a PVM and therefore the projectors in $\widetilde{C}$ satisfy the conditions for a \textbf{WA-PVM} for $\mathcal{S}$.

To show the reverse implication. we demonstrate that given a weakly antidistinguishing PVM $E=\{e_i\}$ for $\mathcal{S}$, we can construct a weakly antidistinguishing rank-1 PVM which in turn allows us to construct a contextuality scenario for which $\mathcal{S}$ is a contextual instance. Let $\Pi=\{\pi_i\}_i$ be a rank-1 PVM associated with fine graining the elements of $E$ such that every non-rank-1 projector $e_k\in E$ is fine-grained into separate rank-1 projectors in $\Pi$, i.e. $e_k = \sum_{\Pi_{e_k}}\pi_l$ where $\Pi_{e_k} \subset \Pi$. Thus we have that all the projectors in $\Pi_{e_{k}}$ are also orthogonal to the same state(s) as $e_k$. Hence $\Pi$ is also a \textbf{WA-PVM} for $\mathcal{S}$. 

Since all projectors in $\Pi$ are orthogonal to at least one state in $\mathcal{S}$, by definition $\Pi \subset \mathcal{I}(\mathcal{S})$, we can construct a context $C_i$ such that $\pi_i, \rho_{\pi_i} \in C_i$ where $\rho_{\pi_i}\in \mathcal{S}$ is an associated orthogonal state to $\pi_i$. 
Therefore we can construct an implied contextuality scenario $\mathcal{G}(X_{\mathcal{S}}, \mathcal{C}_{\mathcal{S}})$ where $X_{\mathcal{S}} =  \left(\cup_i C_i\right) \cup \mathcal{S}$ and $\mathcal{C}_{\mathcal{S}}=\mathbb{C}(X_{\mathcal{S}})$. Now we attempt to find a noncontextual value assignment $\nu$ for $\mathcal{G}(X_{\mathcal{S}}, \mathcal{C}_{\mathcal{S}})$. However if $\nu$ satisfies the condition that $\nu(\sigma)=1$ for all $\sigma \in \mathcal{S}$, then it cannot simultaneously satisfy  $\sum_{\sigma\in C_{i}}\nu(\sigma)=1$ and $\sum_{\pi_{i}\in \Pi}\nu(\pi_{i})=1$. Since $C_i, \Pi \in \mathcal{C}_{\mathcal{S}}$, we have that $\mathcal{S}$ is a contextual instance for $\mathcal{G}(X_{\mathcal{S}}, \mathcal{C}_{\mathcal{S}})$.
\end{proof}

\section{Proof of Lemma \ref{lemma:maxcontextual_pvm}}\label{app:lemma1}

\noindent
\textbf{Lemma 1.} \textit{Any PVM $\Pi=\{\pi_i\}$ which weakly antidistinguishes a critically contextual set $\mathcal {S}$ must also strongly antidistinguish it.}

\begin{proof}
Let $S_{i}\subset \mathcal{S}$ such that $\Tr\left(\pi_i\rho_k\right)=0$, $\forall \rho_k\in S_i$. We will first prove by contradiction that $\cup_i S_i=\mathcal{S}$, which will show that all states are excluded by some outcome. For this we assume that there exists $\rho\in \mathcal{S}\backslash \{\cup_i S_i\}$. One can remove $\rho$ from $\mathcal{S}$ and the remaining set, that is $\mathcal{S}\backslash\{\rho\}$, is still weakly antidistinguishable by the $\Pi$. This implies from Theorem \ref{theorem:main_proj} that  $\mathcal{S}\backslash\{\rho\}$ is contextual, which in turn leads to a contradiction to $\mathcal{S}$ being critically contextual.
Thus we have that $\cup_i S_i=\mathcal{S}$, which implies that for every $\rho_k\in \mathcal{S}$, there exists $\pi_k \in \Pi$, such that $\Tr\left(\pi_k\rho_k\right)=0$. Now we assume that $\Pi$ is NOT a strongly antidistinguishing PVM for $\mathcal{S}$ which implies there exists $\rho\in \mathcal{S}$ such that for every $\pi_k \in \Pi$ that is orthogonal to $\rho$, there must also exist another state $\rho_k$ such that $\Tr\left(\pi_k\rho_k\right)=0$. Thus $\mathcal{S}/\{\rho\}$ remains weakly antidistinguishable by $\Pi$, which  again contradicts $\mathcal{S}$ being critically contextual. 
\end{proof}

\section{Critical contextuality of PBR states}\label{app:pbr}

Here we show that the \textit{PBR states} \cite{10.1038/nphys2309}, $\mathcal{S}_{\mathrm{PBR}}= \{\ket{0}\otimes\ket{0}, \ket{0}\otimes\ket{+}, \ket{+}\otimes\ket{0}, \ket{+}\otimes\ket{+}\}$, are maximally contextual. It was shown in \cite{10.1038/nphys2309} that $\mathcal{S}_{\mathrm{PBR}}$ has an \textbf{A-PVM} which projects onto the following states, 
\begin{align*}
    \ket{\xi_1} &= (\ket{0}\otimes\ket{1}+ \ket{1}\otimes\ket{0})/\sqrt{2}  \\
    \ket{\xi_2} &= (\ket{0}\otimes\ket{-} + \ket{1}\otimes\ket{+})/\sqrt{2}  \\
    \ket{\xi_3} &= (\ket{+}\otimes\ket{1} + \ket{-}\otimes\ket{0})/\sqrt{2}  \\
    \ket{\xi_4} &= (\ket{+}\otimes\ket{-} + \ket{-}\otimes\ket{+})/\sqrt{2}  
\end{align*}
In fact, a closer look reveals that this measurement is also an \textbf{SA-PVM}, as every outcome exclusively excludes a state and no other state in the set. From Theorem \ref{theorem:main_proj} we know that $\mathcal{S}_\mathrm{PBR}$ is contextual and we now show that it is also maximally contextual by utilising the following two Lemmas that are specific for a set of three states: 
\begin{lemma}\label{lemma:three_states}
    All non-orthogonal sets of three states that are \textbf{WA}, are necessarily \textbf{A} and \textbf{SA}.
\end{lemma}
\begin{proof}
    This can be shown by contradiction. If we assume a set of three states $\mathcal{S}$ is \textbf{WA} and not \textbf{SA}, then by Theorem \ref{theorem:max_cont_antidist}, $\mathcal{S}$ cannot be maximally contextual. Hence there must exist a maximally contextual subset $\mathcal{S}'\subset\mathcal{S}$ over two states that is \textbf{SA} by Theorem \ref{theorem:max_cont_antidist}. For a set of two states, an \textbf{SA-PVM} is able to distinguish the states in $\mathcal{S}'$. Since non-orthogonal states cannot be distinguished, this is a contradiction. Thus the set must be \textbf{SA} and therefore also \textbf{A}.
\end{proof}
\begin{lemma}\label{lemma:cond_three_states_anti} (Ref. \cite{10.1103/physreva.66.062111})
    Given a set of three states $\mathcal{S}=\{\psi_1, \psi_2, \psi_3\}$, let their pairwise overlaps be $\delta_1 = |\braket{\psi_1|\psi_2}|^2$, $\delta_2 = |\braket{\psi_1|\psi_3}|^2$, $\delta_3 = |\braket{\psi_2|\psi_3}|^2$. The set $\mathcal{S}$ is antidistinguishable if and only if both the following conditions hold,
    \begin{align}
        \delta_1 + \delta_2 + \delta_3 &< 1 \label{eq:condition_pbr}\\
        (\delta_1 + \delta_2 + \delta_3 - 1)^2 &\geq 4\delta_1 \delta_2 \delta_3.
    \end{align}
\end{lemma}
\begin{proof}
    See Ref. \cite{10.1103/physreva.66.062111} for proof
\end{proof}

Pairwise overlaps of states in $\mathcal{S}_{\mathrm{PBR}}$ are either $\frac{1}{2}$ or $\frac{1}{4}$. It can be seen through enumeration over all possible subsets $\mathcal{S}\subset\mathcal{S}_\mathrm{PBR}$ of size three, there is no such subset that satisfies Eq. \eqref{eq:condition_pbr}. Hence there are no subsets of three states that are antidistinguishable. By Lemma \ref{lemma:three_states} and Theorem \ref{theorem:main_proj}, we know that therefore there are no subsets of three states that are contextual. Since subsets of two non-orthogonal states cannot be contextual, we have shown that there are no contextual subsets of $\mathcal{S}_{\mathrm{PBR}}$, implying that $\mathcal{S}_{\mathrm{PBR}}$ is maximally contextual.

\newpage

\begin{widetext}

\section{Supplementary Material -- Examples}\label{sec:examples}

We provide two examples of strongly antidistinguishable sets that are critically contextual.

\subsection{ Lison\v ek KS Set}\label{sec:lisonek}

\begin{figure*}
    \centering
    \includegraphics[width=460pt]{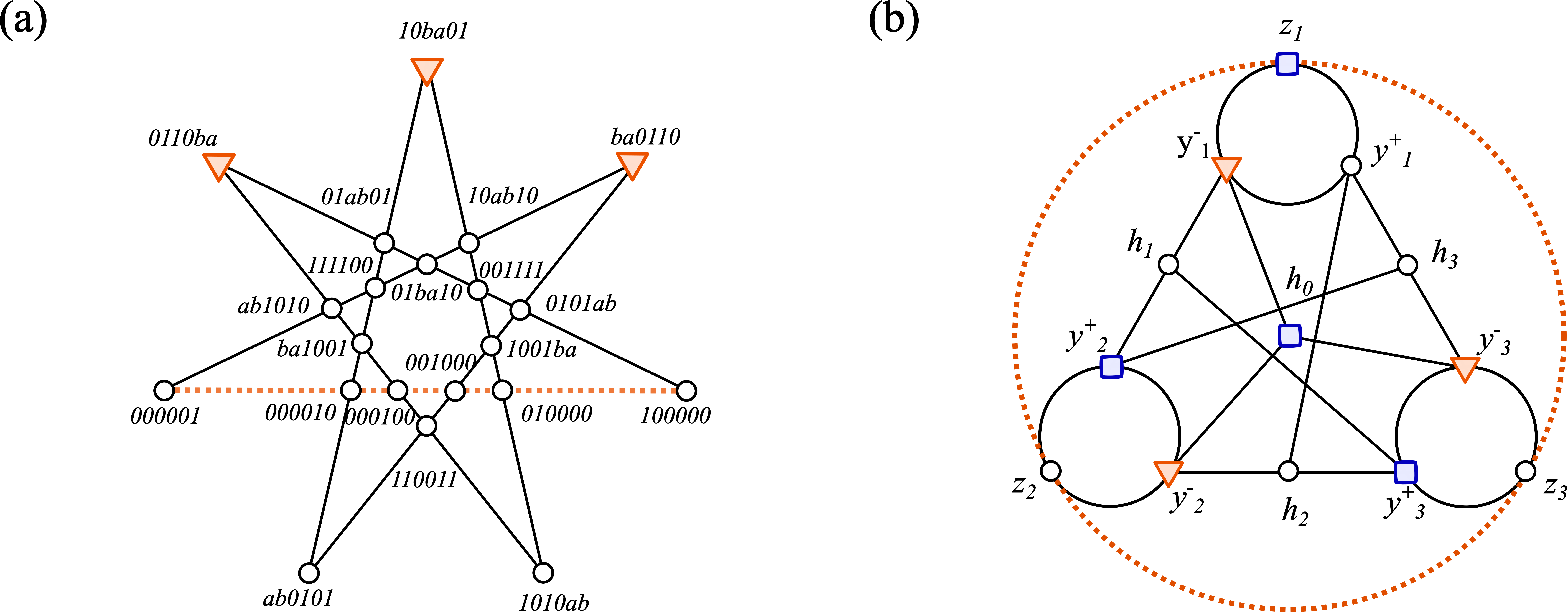}
    \vspace{0.2cm}
    \caption{(a) The Lison\v{e}k state-independent proof of contextuality \cite{PhysRevA.89.042101} with 21 projectors and 7 contexts in dimension $d=6$ is the smallest KS set in terms of contexts. The contexts here are given by straight lines and the labelling $10ab10$ denotes the projector associated with the vector $(1,0,a,b,1,0)$, where $a=e^{2\pi i/3}$ and $b=e^{4\pi i/3}$. The three projectors indicated by orange triangles generate the scenario with the associated antidistinguishing context given by the orange, dotted line. (b) Yu and Oh's set \cite{10.1103/physrevlett.108.030402} is the simplest state independent proof of contextuality with the smallest number of projectors in dimension $d=3$. See section \ref{sec:yu_oh} specifying the projectors. Contexts in this figure are illustrated as circles. The orange projectors generate a contextual instance with the \textbf{SA} measurement given by the orange, dotted line. However, the scenario is noncontextual by our definition of contextuality (Def. 3) as it admits a consistent value assignment highlighted by the blue squares.}
    \label{fig:yu_oh_proof_and_lisonek}
\end{figure*}

We look at the contextuality scenario presented in Figure \ref{fig:yu_oh_proof_and_lisonek}(a) constructed by Lison\ifmmode \check{e}\else \v{e}\fi{}k et al. \cite{PhysRevA.89.042101}, which has 21 projectors and 7 contexts. The selection of any set of projectors that generate the Lison\v{e}k set is guaranteed to have a weakly antidistinguishing measurement by Theorem 1. In fact it is not possible to have greater than 3 projectors generating the set -- as there are 7 contexts and each projector belongs to two -- and no less than three. A generating set of three is represented in orange in Figure \ref{fig:yu_oh_proof_and_lisonek}. There exists an \textbf{SA-PVM} for this set, represented by the orange dotted line. From lemma 2, we know that this set of states is critically contextual.

\subsection{Yu-Oh Set}\label{sec:yu_oh}

We now present an example that reveals some subtlety in the way that we have defined contextuality. 
Consider the set of non-orthogonal states $\mathcal{S}=\{ y_1^- , y_2^+ , y^-_3\}$ and the implied set $I=\{y_1^+ ,y_3^+ , y_2^- , h_0, h_1 , h_2, h_3 , z_1 , z_2, z_3\} \subset \mathcal{I}(S)$ where \cite{10.1103/physrevlett.108.030402},
\begin{alignat*}{4}
    & h_1 = P(-1,1,1) \quad\quad && y_1^- = P(0,1,-1) \quad\quad && y_1^+ = P(0,1,1)  \quad\quad && z_1 = P(1,0,0) \\
    & h_2 = P(1,-1,1) \quad\quad && y_2^- = P(-1,0,1) \quad\quad && y_2^+ = P(1,0,1)  \quad\quad && z_2 = P(0,1,0) \\
    & h_3 = P(1,1,-1) \quad\quad && y_3^- = P(1,-1,0) \quad\quad && y_3^+ = P(1,1,0)  \quad\quad && z_3 = P(0,0,1) \\
    & h_0 = P(1,1,1) \quad\quad && \quad\quad &&  \quad\quad && 
\end{alignat*}
with $P(a,b,c):=\ket{a,b,c}\bra{a,b,c}$ such that $\ket{a,b,c}:= a\ket{0} + b\ket{1} + c\ket{2}$. The set of projectors $X:=I \cup \mathcal{S}$ and the associated contexts illustrated in Figure \ref{fig:yu_oh_proof_and_lisonek}(b), form the Yu and Oh's set \cite{10.1103/physrevlett.108.030402}. While this well-known set does not fit the definition of a state-independent proof of contextuality as presented in this paper, there exists a more general definition of state-independent contextuality \cite{RevModPhys.94.045007}, for which this is an example. It is in fact an example of the more general proof of state-independent contextuality with the smallest number of  contexts possible.

This is an interesting example, as it is an example of when a non-contextual scenario allows a contextual instance. The contextuality scenario in Figure \ref{fig:yu_oh_proof_and_lisonek}(b) does admit a noncontextual value assignment (with projectors coloured blue assigned to 1 and all other projectors assigned zero) and the scenario is therefore noncontextual by our definition. Nevertheless, we have a critically contextual set given by the orange projectors that generate the scenario, with the \textbf{SA-PVM} highlighted by the context in orange in Figure \ref{fig:yu_oh_proof_and_lisonek}(b).

Not only can contextuality be observed in non-deterministic settings (as opposed to the binary assignment of values in Def. 3, but the subtlety in \cite{10.1103/physrevlett.108.030402} highlights that the algebraic conditions of Kochen and Specker do not entirely characterise all types of contextuality. One future direction is to recast our characterisation of contextual sets of states with a more general algebraic framework~\cite{10.48550/arxiv.2408.16764}. 

\end{widetext}

\end{document}